\documentclass[letter]{ieice}
\usepackage[dvipdfmx]{graphicx,xcolor}
\usepackage[fleqn]{amsmath}
\usepackage{newtxtext}
\usepackage[varg]{newtxmath}

\usepackage{mdwlist}
\usepackage{amssymb} 

\usepackage{amsthm}
\usepackage{array}
\usepackage{enumerate}
\usepackage{enumitem}
\usepackage{enumitem}
\setenumerate[1]{itemsep=0pt,partopsep=0pt,parsep=\parskip,topsep=5pt}
\setitemize[1]{itemsep=0pt,partopsep=0pt,parsep=\parskip,topsep=5pt}
\setdescription{itemsep=0pt,partopsep=0pt,parsep=\parskip,topsep=5pt}

\usepackage{longtable}
\usepackage{graphicx}
\usepackage{galois}
\usepackage[CJKbookmarks=true,colorlinks,linkcolor=black]{hyperref}
\usepackage[justification=centering]{caption}
\usepackage[square,sort,comma,numbers]{natbib}
\let\biblatexRN\RN
\protected\def\RN{\ifmmode\bcR^{\bcN}\else\expandafter\biblatexRN\fi}

\makeatletter
\g@addto@macro\normalsize{%
    \setlength\abovedisplayskip{1pt}
    \setlength\belowdisplayskip{1pt}
    \setlength\abovedisplayshortskip{1pt}
    \setlength\belowdisplayshortskip{1pt}
}

\newcommand{\Z}{\mathbb{Z}}
\newcommand{\PZ}{\mathbb{Z}^{+}}
\newcommand{\F}{\mathbb{F}}
\newcommand{\N}{\mathbb{N}}

\DeclareMathOperator{\Ima}{Im}
\makeatother
\setlist[enumerate]{label*=\arabic*.}

\theoremstyle{definition}

\theoremstyle{plain}
\newtheorem{thms}{Theorem}
\newtheorem{lemma}{Lemma}
\newtheorem{prop}{Proposition}
\newtheorem{conc}{Corollary}

\theoremstyle{remark}
\newtheorem{remark}{Remark}

\field{A}
\title[A Note on Two Constructions of ZDB functions]{A Note on Two Constructions of Zero-Difference Balanced Functions}
\titlenote{The work of was supported in part by the National Science Foundation of China (Grant No. 61502113, 61602125, 61702124) and the Guangxi Provincial National Science Foundation (Grant No. 2016GXNSFBA380153), and the China Postdoctoral Science Foundation (Grant No. 2018M633041).}

\authorlist{%
 \authorentry[tpu01yzx@gmail.com]{Zongxiang Yi}{m}{GZHU}\MembershipNumber{111111}
 \authorentry[yuyuyin@163.com]{Yuyin Yu}{n}{GZHU}
 \authorentry[ctang@gzhu.edu.cn]{Chunming Tang}{n}{GZHU}
 \authorentry[zhengyanbin@guet.edu.cn]{Yanbin Zheng}{n}{GUET}
}
\affiliate[GZHU]{The authors are with the School of Mathematics and Information Science, Guangzhou University, Guangzhou, 510006, P.R. China, and the Key Laboratory of Information Security, School of Mathematics and Information Science, Guangzhou University, Guangzhou, 510006, P.R. China.
}
\affiliate[GUET]{The author is with the School of Computer Science and Network Security, Dongguan University of Technology, Dongguan 523808, China, and the Guangxi Key Laboratory of Cryptography and Information Security, Guilin University of Electronic Technology, Guilin 541004, China.
}

\received{2018}{1}{1}
\revised{2018}{1}{1}

\begin{document}

\maketitle
\begin{summary}
Notes on two constructions of zero-difference balanced (ZDB) functions are made in this letter. Then ZDB functions over $\Z_{e}\times \prod_{i=0}^{k}{\F_{q_i}}$ are obtained. And it shows that all the known ZDB functions using cyclotomic cosets over $\Z_{n}$ are special cases of a generic construction. Moreover, applications of these ZDB functions are presented.
\end{summary}

\section{Introduction}\label{se:intro}
Let $(A,+)$ and $(B, +)$ be two finite abelian groups. A function $f$ from $A$ to $B$ is called an $(n,m,\lambda)$ zero-difference balanced (ZDB) function if there is a constant number $\lambda$, such that
$$|\{ x\in A \mid f(x+a) - f(x) = 0\}| = \lambda$$
for every element $a \in A\setminus\{0\}$, where $n=|A|$, $m=|\Ima(f)|$, and $\Ima(f)$ is the image set of $f$.

In \citeyear{ding2008optimal}, \citeauthor{ding2008optimal} first proposed the concept of ZDB function and showed that optimal constant composition codes (CCC) can be obtained from ZDB functions \cite{ding2008optimal}. Later, \citeauthor{ding2009optimal}\cite{ding2009optimal}, \citeauthor{zhou2012some}\cite{zhou2012some} and \citeauthor{wang2014sets}\cite{wang2014sets} showed that optimal perfect  difference systems of sets (DSS), optimal constant weight codes (CWC) and optimal frequency-hopping sequences (FHS) can be obtained from ZDB functions, respectively. Since CCC, CWC, DSS and FHS have many applications in combination designs and communication systems, many researchers have been working on constructing more ZDB functions (see \cite{ding2009optimal, ding2008optimal, carlet2004highly, zhou2012some, wang2014sets, cai2013new, ding2014three, zha2015cyclotomic, zhifan2015zero, cai2017zero, li2017generic, yi2017generic} and the references therein).

In this letter, we concern those ZDB functions constructed by generalized cyclotomic cosets. Some authors \cite{ding2014three, zha2015cyclotomic} studied the construction of ZDB functions on the rings $\Z_n$. They showed that there exist non-trivial ZDB functions on $\Z_n$ only for odd integers $n$. \citeauthor{ding2014three} \cite{ding2014three} constructed a class of ZDB functions for any positive integer $n$ by using the product of finite fields. In \citeyear{yi2017generic}, \citeauthor{yi2017generic} generalized the construction from the residue rings $\Z_n$ and finite fields $\F_q$ to generic rings \cite{yi2017generic}.

The main contribution of the letter are twofold. Firstly, by generalizing the construction of ZDB functions proposed by \citeauthor{cai2017zero} \cite{cai2017zero}, ZDB functions can be obtained over $\Z_{e}\times \prod_{i=0}^{k}{\F_{q_i}}$. Secondly, it shows that all the known ZDB functions using cyclotomic cosets over $\Z_{n}$ \cite{cai2013new, ding2014three, zha2015cyclotomic}, are indeed special cases of the generic construction in \cite{yi2017generic}.

 This letter is organized as follows: In Section \ref{se:new}, the construction in \cite{yi2017generic} is recalled, and then notes on two constructions of ZDB functions are made. Applications of ZDB functions are presented in Section \ref{se:app}. Section \ref{se:con} concludes this letter.

\section{Two constructions of ZDB Functions}\label{se:new}

\subsection{Notations}
Unless otherwise stated, $(R, +, \times)$ is always a commutative ring with identity. Let $R^\times$ denote the set of all invertible elements in $(R, \times)$. Let $R^*$ denote the set of all nonzero elements in $R$. Define $x/y=x\times y^{-1}$, for $x \in R, y\in R^\times$.

For any subset $A$ of $R$ and any element $a$ of $R$, define
\begin{equation*}\begin{split}
 a+A=\{a+x \mid x \in A\},  A+a=a+A \\
 aA=\{ax \mid x \in A\},  Aa=\{xa \mid x \in A\}.
\end{split}\end{equation*}

Moreover, the set of all natural numbers is denoted by $\N$. The set of all integers is denoted by $\Z$. The set of all positive integers is denoted by $\PZ$. A finite field with $q$ elements is denoted by $\F_q$.

\subsection{The Method of Yi}\label{sse:yi}
In this subsection, we will recall the method of Yi \cite{yi2017generic}.
\begin{prop}[\cite{yi2017generic}]\label{prop:yi_generic_construction}
Let $(R,+,\times)$ be a ring of order $n$, and let $G$ be a subgroup of $(R,\times)$. Suppose $|G|=e$. Define $\mathbb{S}=\{\alpha G \mid \alpha \in R \}$. If $G$ satisfies the condition
\begin{equation}\label{eq:condition_over_generic_ring}(G-1)\setminus\{0\}\subset R^\times,\end{equation}
then
\begin{enumerate}
\item $\mathbb{S}$ is a partition of $R$;
\item $|\alpha G|=e$, for every $\alpha \in R^*$;
\item $|\mathbb{S}|=\frac{n-1}{e}+1$;
\item For every $a \in R^*$, \begin{equation*}\begin{split}\label{eq:solutions_of_ZDB}
\{x \in R \mid f(x+a)=f(x)\} = \{a(g-1)^{-1} \mid g \in G\setminus\{1\}\};
\end{split}\end{equation*}
\item $f=f_2(f_1(x))$ is an $(n,\frac{n-1}{e}+1,e-1)$ ZDB function from $(R,+)$ to $(\Z_{\frac{n-1}{e}+1}, +)$, where $f_1(x)$ is the map from $R$ to $\mathbb{S}$ which maps an element $x$ into $\alpha G$ such that $x \in \alpha G$,  and $f_2(x)$ is an arbitrary bijective map from $\mathbb{S}$ to $\Z_{\frac{n-1}{m}+1}$.
\end{enumerate}
\end{prop}

In Proposition \ref{prop:yi_generic_construction}, a set of coset representatives of $\mathbb{S}$, denoted by $L_G$, can be obtained by randomly selecting one element in $\alpha G \in \mathbb{S}$. Note that $0G=\{0\} \in \mathbb{S}$ and $0 \in L_G$. Let $\mathbb{S}^*=\mathbb{S}\setminus\{\{0\}\}$. Similarly, a set of coset representatives of $\mathbb{S}^*$, denoted by $L_G^*$, can be obtained too, namely, $L_G^*=L_G\setminus\{0\}$.

\subsection{One construction of ZDB functions}\label{se:new_zdb1}
In this subsection, we will propose one construction of ZDB functions. With the notations in Subsection \ref{sse:yi}, two indicators are defined. For any $r \in R^*$, there exists a unique element $\alpha \in L_G^*$ such that $r \in \alpha G$. Furthermore, there exists a unique element $g \in G$ such that $r = \alpha g$.

Now the row indicator $RI_{L_G^*}$ and the column indicator $CI_{L_G^*}$ are defined as follows:
\begin{equation*}\begin{split}
RI_{L_G^*}:& R^* &\to &L_G^*, & CI_{L_G^*}:& R^*& \to &G,\\
& r&\mapsto &\alpha, & &r& \mapsto &g.
\end{split}\end{equation*}
The column indicator has the following property.

\begin{lemma}\label{lm:property_A_of_CI}
Let $(R,+,\times)$ be a ring. Let $G$ be an subgroup of $(R, \times)$, such that $(G-1)\setminus\{0\}\subset R^\times$. Suppose $e=|G|$. Let $a(g)$ be a function from $G$ to $R^*$, such that $RI_{L_G^*}(a(g) g)=RI_{L_G^*}(a(g))$ for every $g \in G$. Then
$$\{ CI_{L_G^*}(a(g)) / CI_{L_G^*}(a(g)g) \mid g \in G \} = G.$$
Moreover, $CI_{L_G^*}(a(g)) / CI_{L_G^*}(a(g)g) = 1$ if and only if $g=1$.
\end{lemma}
\begin{proof}
Suppose $RI_{L_G^*}(a(g))=\alpha$. Let $a(g)=\alpha g'$, where $g'\in G$. Then $a(g)g = \alpha g'g$. We have
$$CI_{L_G^*}(a(g)) / CI_{L_G^*}(a(g)g)=g' / (g'g).$$
When $g$ runs over $G$, $g'/(g'g)$ run over $G$ too. Note that $g'/(g'g)=1$ if and only if $g=1$ for any $g' \in G$. So $CI_{L_G^*}(a(g)) / CI_{L_G^*}(a(g)g) = 1$ if and only if $g=1$.
\end{proof}

Now we give our construction as Theorem \ref{th:new_zdb_1}.
\begin{thms}\label{th:new_zdb_1}
Let $(R,+,\times)$ be a ring of order $n$, and let $G$, $H$ be two subgroups of $(R,\times)$. Suppose the following conditions hold:
\begin{enumerate}
\item $(G-1)\setminus\{0\}\subset R^\times$;
\item $(H-1)\setminus\{0\}\subset R^\times$;
\item $|H|=|G|-1$.
\end{enumerate}
Then there exist $(en,\frac{en-1}{e-1}+1,e-2)$ ZDB functions from $(R, +) \times (G, \times)$ to $(\Z_{\frac{en-1}{e-1}+1}, +)$, where $e=|G|$.
\end{thms}
\begin{proof}
Let $0$ and $1$ denote the identities of $(R,+)$ and $(R,\times)$, respectively. Define
\begin{equation}\label{eq:define_of_T1}
\mathbb{T}=\{ 0, (0,1) \} \bigcup L_H^* \bigcup L_G^* \times G.
\end{equation}
Followed from Proposition \ref{prop:yi_generic_construction}, we have $|L_H^*|=\frac{n-1}{e-1}$, $|L_G^*|=\frac{n-1}{e}$. Thus $|\mathbb{T}|=|L_H^*|+|L_G^*|e+2=\frac{en-1}{e-1}+1$. Denote $\overline{R}=(R, +) \times (G, \times)$. Now we define a function from $\overline{R}$ to $\mathbb{T}$:
\begin{equation*}\begin{split}
f_1(r,x)=&\begin{cases}
0, & \text{if $r=0$ and $x=1$,} \\
(0,1), & \text{if $r=0$ and $x\ne 1$,} \\
RI_{L_H^*}(r), & \text{if $r \ne 0$ and $x=1$,} \\
(RI_{L_G^*}(r), xCI_{L_G^*}(r)), & \text{if $r \ne 0$ and $x\ne1$.}
\end{cases}
\end{split}\end{equation*}
Let $f_2(x)$ be an arbitrary bijective map from $\mathbb{T}$ to $\Z_{\frac{en-1}{e-1}+1}$. We assert that $f=f_2(f_1(x))$ is an $(en,\frac{en-1}{e-1}+1,e-2)$ ZDB function from $\overline{R}$ to $(\Z_{\frac{en-1}{e-1}+1}, +)$. Obviously, for any $\Delta=(\Delta_r, \Delta_x) \ne (0, 1)$, we have
$$|\{y \in \overline{R} \mid f(y+\Delta)-f(y)=0\}|=|\{y\in \overline{R} \mid f_1(y+\Delta)=f_1(y)\}|. $$
In the following, we will show that
$$|\{y\in \overline{R} \mid f_1(y+\Delta)=f_1(y)\}|=e-2.$$

Firstly, we make a partition of $\overline{R}$. Let
$$\overline{R}= \bigcup_{i=1}^{4}{R_i},$$
where
$$R_1=\{(r, x) \in \overline{R} \mid r=0, x=1\}, R_2=\{(r, x) \in \overline{R} \mid r=0, x\ne 1\},$$
$$R_3=\{(r, x) \in \overline{R} \mid r\ne0, x=1\}, R_4=\{(r, x) \in \overline{R} \mid r\ne0, x\ne 1\}.$$
Note that if $f_1(y+\Delta)=f_1(y)$, then $y+\Delta$ and $y$ must belong to some $R_i$ where $1\le i \le 4$.

Secondly, we have a discussion over $(\Delta_r, \Delta_x)\ne (0, 1)$.
\begin{enumerate}
\item Case $\Delta_r \ne 0$ and $\Delta_x = 1$:
    \begin{enumerate}
    \item If $(r,x) \in R_1$, then $(\Delta_r, 1) \notin R_1$. So
    $$|\{ (r, x) \in R_1 \mid f_1(\Delta_r, 1)=f_1(0,1)\}| = 0.$$
    \item If $(r,x) \in R_2$, then $(\Delta_r, x) \notin R_2$. So
    $$|\{ (r,x) \in R_2 \mid f_1(\Delta_r, x)=f_1(0,x)\}| = 0.$$
    \item If $(r,x) \in R_3$, then
    \begin{equation*}\begin{split}
     & | \{ (r, x) \in R_3 \mid f_1(r + \Delta_r, 1)=f_1(r,1)\} |\\
    =& | \{ r \in R \mid \text{ $RI_{L_H^*}(r + \Delta_r) =  RI_{L_H^*}(r)$ } | \\
    =& |H|-1 = e-2.
    \end{split}\end{equation*}
    In the above, the second identity is followed from Proposition \ref{prop:yi_generic_construction}.
    \item If $(r,x) \in R_4$, then $r + \Delta_r \ne r$. It implies either $RI_{L_H^*}(r+\Delta_r) \ne RI_{L_H^*}(r)$  or $xCI_{L_G}(r+\Delta_r)\ne xCI_{L_G}(r)$. Both of them would lead to $f_1(r + \Delta_r, x) \ne f_1(r,x)$. So
    $$|\{ (r,x) \in R_4 \mid f_1(r + \Delta_r, x)=f_1(r,x)\}| = 0.$$
    \end{enumerate}
    To sum up, when $\Delta_r \ne 0$ and $\Delta_x = 1$, we have
    $$|\{ (r,x) \in \overline{R} \mid f_1(r + \Delta_r, x)=f_1(r,x)\}| = e-2.$$
\item Case $\Delta_r = 0$ and $\Delta_x \ne 1$:
    \begin{enumerate}
    \item If $(r,x) \in R_1$, then $(0, \Delta_x) \notin R_1$. So
    $$|\{ (r,x) \in R_1 \mid f_1(0, \Delta_x)=f_1(0,1)\}| = 0.$$
    \item If $(r,x) \in R_2$, then $f_1(0, x\Delta_x)=f_1(0,x)$, if and only if, both $x\Delta_x\ne 1$ and $x\ne 1$ hold. So
    $$|\{ (r,x) \in R_2 \mid f_1(0, x\Delta_x)=f_1(0,x)\}| = e-2.$$
    \item If $(r,x) \in R_3$, then $(r, \Delta_x) \notin R_3$. So
    $$| \{ (r, x) \in R_3 \mid f_1(r, \Delta_x)=f_1(r,1)\} | = 0.$$
    \item If $(r,x) \in R_4$, then $x \Delta_x CI_{L_G}(r) \ne x CI_{L_G}(r) $. So
    $$|\{ (r,x) \in R_4 \mid f_1(r, x \Delta_x)=f_1(r,x)\}| = 0.$$
    \end{enumerate}
    To sum up, when $\Delta_r = 0$ and $\Delta_x \ne 1$, we have
    $$|\{ (r,x) \in \overline{R} \mid f_1(r + \Delta_r, x)=f_1(r,x)\}| = e-2.$$
\item Case $\Delta_r \ne 0$ and $\Delta_x \ne 1$:
    \begin{enumerate}
    \item If $(r,x) \in R_1$, then $(\Delta_r, \Delta_x) \notin R_1$. So
    $$|\{ (r,x) \in R_1 \mid f_1(\Delta_r, \Delta_x)=f_1(0,1)\}| = 0.$$
    \item If $(r,x) \in R_2$, then $(\Delta_r, x \Delta_x) \notin R_2$. So
    $$|\{ (r,x) \in R_2 \mid f_1(\Delta_r, x\Delta_x)=f_1(0,x)\}| = 0.$$
    \item If $(r,x) \in R_3$, then $(r + \Delta_r, \Delta_x) \notin R_3$. So
    $$| \{ (r, x) \in R_3 \mid f_1(r + \Delta_r, \Delta_x)=f_1(r,1)\} | = 0.$$
    \item If $(r,x) \in R_4$, then
    \begin{equation*}\begin{split}
     & |\{ (r,x) \in R_4 \mid f_1(r + \Delta_r, x \Delta_x)=f_1(r,x)\}|\\
    =& |\{ (r,x) \in R_4 \mid \begin{array}{c}
RI_{L_G^*}(r+\Delta_r) = RI_{L_G^*}(r),\\
x \Delta_x CI_{L_G}(r+\Delta_r) = x CI_{L_G}(r) , \\
x \ne 1 \text{ and } x \Delta_x \ne 1
\end{array} \} | \\
    =& |\{ (r,x) \in R_4 \mid \begin{array}{c}
r=\Delta_r(g-1)^{-1}, g \in G\setminus\{1\},\\
\Delta_x I_{L_G}(rg) = I_{L_G}(r) ,\\
x \ne 1 \text{ and } x \Delta_x \ne 1
\end{array} \} | \\
    =& |\{ r \in R \mid \begin{array}{c}
r=\Delta_r(g-1)^{-1}, g \in G\setminus\{1\},\\
\text{ and } \Delta_x = CI_{L_G}(r) / CI_{L_G}(rg)
\end{array} \} | \\
 & \times |\{ x \in G \mid x \ne 1 \text{ and } x \Delta_x \ne 1. \} | = e-2.
    \end{split}\end{equation*}
\end{enumerate}
The second identity is followed from Proposition \ref{prop:yi_generic_construction}, and the last identity is followed from Lemma \ref{lm:property_A_of_CI}.

To sum up, when $\Delta_r \ne 0$ and $\Delta_x \ne 1$, we have
    $$|\{ (r,x) \in \overline{R} \mid f_1(r + \Delta_r, x \Delta_x)=f_1(r,x)\}| = e-2.$$
\end{enumerate}
Finally, when $\Delta=(\Delta_r, \Delta_x) \ne (0, 1)$, we have
$$|\{y\in \overline{R} \mid f_1(y+\Delta)=f_1(y)\}|=e-2.$$
\end{proof}
\begin{remark}
Theorem \ref{th:new_zdb_1} can also be obtained by the main construction in \cite{li2017generic}, but the conditions that Theorem \ref{th:new_zdb_1} requires are much simpler to be considered and easier to be checked than those in \cite{li2017generic}.
\end{remark}

To apply Theorem \ref{th:new_zdb_1}, let $R=\Z_n$. In \cite{zha2015cyclotomic} the authors have shown how to construct subgroups satisfying the Condition \eqref{eq:condition_over_generic_ring} on $\Z_n$. So we have the following result.
\begin{conc}\label{conc:zdb_1_on_zn}
Let $n=p_1^{r_1}p_2^{r_2}\cdots p_k^{r_k}$, where $2<p_1<p_2<\cdots < p_k$ are odd prime numbers, and $r_1, r_2, \ldots, r_k$ are positive integers. Then for any positive integers $e$ such that $e(e-1)\mid \gcd(p_1-1, p_2-1, \cdots, p_k-1)$, there exist $(en,\frac{en-1}{e-1}+1,e-2)$ ZDB functions from $(\Z_{en},+)$ to $(\Z_{\frac{en-1}{e-1}+1},+)$.
\end{conc}
\begin{remark}
Corollary \ref{conc:zdb_1_on_zn} is the same as Theorem 1 in \cite{cai2017zero}. So Theorem \ref{th:new_zdb_1} in this paper can be viewed as a generalization of Theorem 1 in \cite{cai2017zero}.
\end{remark}

Moreover, we can obtain ZDB functions over the product of some finite fields by Theorem \ref{th:new_zdb_1}. Note that $F_q^\times$ is cyclic and any cyclic group is isomorphic to $\Z_e$ for some integer $e$. So in Corollary \ref{conc:zdb_1_on_fq}, when applying Theorem \ref{th:new_zdb_1}, we use $(\Z_e, +)$ instead of$(G, \times)$ where $e=|G|$.
\begin{conc}\label{conc:zdb_1_on_fq}
Let $n=p_1^{r_1}p_2^{r_2}\cdots p_k^{r_k}$, where $p_1<p_2<\cdots < p_k$ are prime numbers, and $r_1, r_2, \ldots, r_k$ are positive integers. Denote $R=\prod_{i=1}^{k}\F_{p_i^{r_i}}$. Then for any positive integer $e$ such that $e(e-1)\mid \gcd(p_1^{r_1}-1, p_2^{r_2}-1, \cdots, p_k^{r_k}-1)$, there exist $(en,\frac{en-1}{e-1}+1,e-2)$ ZDB functions from $(R\times \Z_e,+)$ to $(\Z_{\frac{en-1}{e-1}+1}, +)$.
\end{conc}
\begin{remark}
\cite{ding2014three} showed how to construct subgroups satisfying the Condition \eqref{eq:condition_over_generic_ring} on $\F_{p_i^{r_i}}$. Then ZDB functions can be obtained by Corollary \ref{conc:zdb_1_on_fq}. For example, let $n=25$ and $e=4$, we obtain a $(100, 34, 2)$ ZDB function from $\F_{25} \times \Z_{4}$ to $\Z_{34}$. Moreover, let $n=121$ and $e=4$, we obtain a $(726, 146, 4)$ ZDB function by Corollary \ref{conc:zdb_1_on_fq}. This ZDB function can not be retrieved by the constructions in \cite{ding2014three, zha2015cyclotomic, cai2017zero}. But it may be retrieved by the construction in \cite{yi2017generic}, if an appropriate ring is given. 
\end{remark}

To show that Theorem \ref{th:new_zdb_1} can generate more ZDB functions over different rings, we consider the matrix ring $M_2(\F_5)$. Denote $A=\left(\begin{array}{cc}3 & 0 \\0&3\\ \end{array}\right)$ and
$B=\left(\begin{array}{cc}4 & 4 \\1&0\\ \end{array}\right)$. It is easy to check that both $G=\langle A \rangle$ and $H=\langle B \rangle$ satsfy the conditions in Proposition \ref{prop:yi_generic_construction}, and that $|G|=4$, $|H|=|G|-1=3$. Hence there exists a $(2500,834,2)$ ZDB function over $M_2(\F_5)$ by Theorem \ref{th:new_zdb_1}. It is the first ZDB function proposed over matrix rings and noncommunicative rings.

\subsection{The other construction of ZDB functions}
In this subsection, we will construct the second construction of ZDB functions. With the notations in Subsection \ref{sse:yi}, we have
\begin{prop}\label{prop:new_zdb_2}
Let $(R,+,\times)$ be a ring of order $n\ge 3$, and let $G$ be a subgroup of $(R,\times)$. If $G$ satisfies the following conditions:
\begin{enumerate}
\item $(G-1)\setminus\{0\}\subset R^\times$;
\item $(G+1)\subset R^\times$,
\end{enumerate}
then there exist $(n,\frac{n-1}{2e}+1,2e-1)$ ZDB functions, where $e=|G|$.
\end{prop}
\begin{proof}
Let $-G=-1\times G$ and $H= G \bigcup (-G)$ is a subgroup of $(R,\times)$. It is easy to verify that $|H|=2e$ and $(H-1)\setminus \{0\} \subset R^\times$. Then the proof is completed by Proposition \ref{prop:yi_generic_construction}.
\end{proof}
\begin{remark}
The proof indicates that Proposition \ref{prop:new_zdb_2} is a special case of Proposition \ref{prop:yi_generic_construction} since the constructed subset $H$ is a subgroup satisfying Condition \eqref{eq:condition_over_generic_ring}. So the special cases of Proposition \ref{prop:new_zdb_2} are indeed special cases of Proposition \ref{prop:yi_generic_construction}.
\end{remark}

Finally, we will illustrate that the ZDB functions in \cite{cai2013new, ding2014three, zha2015cyclotomic} are indeed special cases of the generic construction in \cite{yi2017generic}.
\begin{enumerate}
\item Let $R=\Z_{n}$ and $G=\langle b \rangle$, where $b$ be an element constructed by Lemma 3 in \cite{zha2015cyclotomic}. Then Theorem 1 in \cite{cai2013new} and Theorem 1 in \cite{zha2015cyclotomic} can be obtained by Proposition \ref{prop:yi_generic_construction}.
\item Let $R=\prod_{i=1}^{k}{\F_{q_i}}$ where $n=\prod_{i=0}^{k}{q_i}$ and $q_i$ are prime powers $(i=1,2,\ldots,k)$. Let $G=\langle b \rangle$, where $b=\langle b_1, b_2, \ldots, b_k \rangle$ and $b_i$ is an element in $\F_{q_i}$ of order $e$ $(i=1,2,\ldots,k)$. Obviously we have $e \mid q_i-1$ for $i=1,2,\ldots,k$. Then Theorem 1 in \cite{ding2014three} can be obtained by Proposition \ref{prop:yi_generic_construction}.
\item Let $R=\Z_{2^m-1}$ and $G=\langle 2 \rangle$, where $m$ is an prime number. It is easy to verify that $G$ satisfies Condition \eqref{eq:condition_over_generic_ring} in Proposition \ref{prop:yi_generic_construction} and $|G|=m$. So Theorem 3 in \cite{ding2014three} can be obtained by Proposition \ref{prop:yi_generic_construction}.
\item Let $R=\Z_{2^m-1}$ and $G=\langle 2 \rangle$, where $m$ is an odd prime number. It is easy to verify that $G$ satisfies all the conditions in Proposition \ref{prop:new_zdb_2} and $|G|=m$. So Theorem 5 in \cite{ding2014three} can be obtained by Proposition \ref{prop:new_zdb_2}.
\item Let $s$ be a prime, $b\ge 2$, and $\gcd(s,b-1)=1$. Let $R=\Z_{\frac{b^s-1}{b-1}}$ and $G=\langle b \rangle$. It is easy to verify that $G$ satisfies Condition \eqref{eq:condition_over_generic_ring} in Proposition \ref{prop:yi_generic_construction} and $|G|=s$. So Corollary 1 in \cite{zha2015cyclotomic} can be obtained by Proposition \ref{prop:yi_generic_construction}.
\item Let $s$ be an odd prime, $b\ge 2$, and $\gcd(s,b-1)=1$. Let $R=\Z_{\frac{b^s-1}{b-1}}$ and $G=\langle b \rangle$. It is easy to verify that $G$ satisfies all the conditions in Proposition \ref{prop:new_zdb_2} and $|G|=s$. So Corollary 2 in \cite{zha2015cyclotomic} can be obtained by Proposition \ref{prop:new_zdb_2}.
\item Let $s$ be a prime, $b\ge 2$, and $\gcd(s,b-1)=1$. Suppose $p=\frac{b^s-1}{b-1}$ is an odd prime. Let $R=\F_{p} \times \F_p$ and $G=\langle b \rangle$. Note that $\Z_{p}=\F_{p}$. It is easy to verify that $G$ satisfies Condition \eqref{eq:condition_over_generic_ring} in Proposition \ref{prop:yi_generic_construction} and $|G|=s$. So Theorem 2 in \cite{zha2015cyclotomic} can be obtained by Proposition \ref{prop:yi_generic_construction}.
\end{enumerate}

\section{Applications}\label{se:app}
The ZDB functions in Proposition \ref{prop:new_zdb_2} have the same structure as those in \cite{yi2017generic}, and they have no new parameters. So only the applications of ZDB functions in Theorem \ref{th:new_zdb_1} are presented in this section. It is necessary to show the following property of our ZDB functions before introducing the applications.

\begin{prop}\label{prop:sizes_of_zdb1}
Let $f:A \rightarrow B$  be an $(en, \frac{en-1}{e-1}+1, e-2)$ ZDB function constructed by Theorem \ref{th:new_zdb_1}, and let $m=\frac{en-1}{e-1}$.  Denote $w_b = |\{ x \in A \mid f(x)=i \}|$ for every $b \in B$ . Then for the multi-set, we have
$$\{ w_b \mid b \in B \}=\{ 1, \underbrace{e-1, e-1, \ldots, e-1}_{m \text{ times}} \}.$$
\end{prop}


\subsection{Optimal Constant Composition Codes}
An $(n, M, d, [w_0, w_1, \ldots, w_{q-1}])_q$ constant composition code (CCC) is a code over an abelian group $\{b_0, b_1, \ldots, b_{q-1}\}$ with length $n$, size $M$ and minimum Hamming distance $d$, such that in every codeword the element $b_i$ appears exactly $w_i$ times for every $i$ ($0\le i\le q-1$). Let $A_q(n, d, [w_0, w_1, \ldots, w_{q-1}])$ denote the maximum size of an $(n, M, d, [w_0, w_1, \ldots, w_{q-1}])_q$ CCC.  A CCC is optimal if the bound in Lemma \ref{lm:bound_of_CCC} is met.
\begin{lemma} \cite{luo2003constant}\label{lm:bound_of_CCC}
If
$$nd-n^2+\sum_{i=0}^{q-1}w_i^2>0,$$
then
$$A_q(n, d, [w_0, w_1, \ldots, w_{q-1}])\le\frac{nd}{nd-n^2+\sum_{i=0}^{q-1}w_i^2}. $$
\end{lemma}

Using the framework in \cite{ding2008optimal}, new optimal CCCs can be constructed from ZDB functions.
\begin{thms}\label{thms:construction_of_CCC}
Let $f$ be an $(en,\frac{en-1}{e-1}+1,e-2)$ ZDB function constructed by Theorem \ref{th:new_zdb_1}. Then there exists an optimal
$(en, en, en-e+2, [1, e-1, e-1, \ldots, e-1])_{\frac{en-1}{e-1}+1}$ CCC.
\end{thms}
To compare the parameters of some known optimal CCCs, the reader is referred to Table \uppercase\expandafter{\romannumeral2} in \cite{cai2017zero}.

\subsection{Optimal Constant Weight Codes}
An $(n, M, d, w)_q$ constant weight code (CWC) is a code over an abelian group $\{b_0, b_1, \ldots, b_{q-1}\}$ with length $n$, size $M$ and minimum Hamming distance $d$, such that the Hamming weight of each codeword is $w$. Let $A_q(n, d, w)$ denote the maximum size of an $(n, M, d, w)_q$ CWC.  A CWC is optimal if the bound in Lemma \ref{lm:bound_of_CWC} is met.
\begin{lemma} \cite{fu1998constructions}\label{lm:bound_of_CWC}
If $nd-2nw+\frac{l}{l-1}w^2>0$, then
$$A_q(n, d, w)\le\frac{nd}{nd-2nw+\frac{l}{l-1}w^2}. $$
\end{lemma}

The codes constructed from ZDB functions in Theorem \ref{thms:construction_of_CCC} are CWCs. \citeauthor{zhou2012some} and \citeauthor{yi2017generic} gave specific constructions in \cite{zhou2012some} and \cite{yi2017generic}, respectively. With the framework established by \cite{yi2017generic}, the ZDB functions constructed in Theorem \ref{th:new_zdb_1} can generate optimal CWCs.

\begin{thms}
With the notations in  Theorem \ref{th:new_zdb_1}, let $f=f_2(f_1(x))$ be an $(en,\frac{en-1}{e-1}+1,e-2)$ ZDB function such that $f_2$ maps $0$ to $0$. Then there exists an optimal $(en, en, en-e+2, en-1)_{\frac{en-1}{e-1}+1}$ CWC.
\end{thms}
\begin{remark}
Note that $f_2(x)$ is a bijective map from $\mathbb{T}$ (defined in \eqref{eq:define_of_T1}) to $\Z_{\frac{en-1}{e-1}+1}$. There are many such bijective maps mapping $0$ to $0$.
\end{remark}

\subsection{Optimal and Perfect Difference Systems of Sets}
Difference systems of sets (DSS) are related with comma-free codes, authentication codes and secrete sharing schemes \cite{Ogata2004New,fuji2009perfect}. Let $\{D_0, D_1, \ldots, D_{q-1}\} $ be disjoint subsets of an abelian group $(G, +)$. Denote $|G|=n$ and $|D_i|=w_i$ for every $i$. Then $\{D_0, D_1, \ldots, D_{q-1}\} $ is said to be an $(n,\{w_0, w_1, \ldots, w_{q-1}\}, \lambda)$ DSS if the multi-set
$$\{x-y \mid \begin{array}{c}
x\in D_i, y\in D_j, 0\le i\ne j \le q-1
\end{array}
\}$$
contains every non-zero element $g \in G$ at least $\lambda$ times. Moreover, a DSS is perfect if every non-zero element $g$ appears exactly $\lambda$ times in the multi-set just mentioned above. It is required that
$$\tau_q(n,\lambda)=\sum_{i=0}^{q-1}|D_i|$$
as small as possible. A DSS is called optimal if the bound in Lemma \ref{lm:bound_of_DSS} is met.
\begin{lemma} \cite{wang2006new}\label{lm:bound_of_DSS}
For an $(n,[w_0, w_1, \ldots, w_{q-1}], \lambda)$ DSS, we have
$$\begin{array}{c}
\tau_q(n,\lambda) \ge \sqrt{SQUARE(\lambda(n-1) + \lceil \frac{\lambda(n-1)}{q-1}\rceil)}
\end{array},$$
where $SQUARE(x)$ denotes the smallest square number that is no less than $x$ and $\lceil x \rceil$ denotes the smallest integer that no less that $x$.
\end{lemma}

Using the framework in \cite{ding2009optimal}, we obtain optimal DSSs in Theorem \ref{th:app_dss}.
\begin{thms}\label{th:app_dss}
Let $f$ be an $(en,\frac{en-1}{e-1}+1,e-2)$ ZDB function constructed by Theorem \ref{th:new_zdb_1}. Then there exists an optimal $(en,\{1,e-1,\ldots,e-1\},en-e+2)$ perfect DSS.
\end{thms}
\begin{remark}
DSSs on non-cyclic groups are related to authentication codes and secret sharing schemes\cite{Ogata2004New, fuji2009perfect}.
\end{remark}
The optimal DSSs constructed in this paper are partitioned-type. To compare the parameters of some known partitioned-type optimal DSSs, the reader is referred to Table \uppercase\expandafter{\romannumeral3} in \cite{cai2017zero}.

\section{Conclusion}\label{se:con}
In this letter, we generalized the construction of ZDB functions in \cite{cai2017zero}. It may instantiate ZDB functions with new parameters if Condition \eqref{eq:condition_over_generic_ring} is studied over other rings. Moreover examples of ZDB functions over noncommunicative ring are first given. Finally we point out that some known ZDB functions are indeed special cases of the generic construction in \cite{yi2017generic}.



\bibliographystyle{ieicetr}
\bibliography{ZDB}

\end{document}